\documentclass{amsart}

\usepackage{graphicx,cite,mathrsfs}
\usepackage{amscd,amsfonts,amsmath,amssymb,amsthm}

\usepackage{hyperref}

\begin{document}

\newtheorem{theorem}{Theorem}[section]
\newtheorem{corollary}[theorem]{Corollary}
\newtheorem{definition}[theorem]{Definition}
\newtheorem{lemma}[theorem]{Lemma}
\newtheorem{proposition}[theorem]{Proposition}
\newtheorem{example}[theorem]{Example}
\theoremstyle{remark}
\newtheorem{remark}{Remark}[section]

\newcommand{\ii}{\mathrm{i}}
\newcommand{\re}{\mathrm{Re}\,}
\newcommand{\im}{\mathrm{Im}\,}

\newcommand{\cD}{\mathcal{D}}
\newcommand{\cH}{\mathcal{H}}
\newcommand{\cL}{\mathcal{L}}
\newcommand{\cW}{\mathcal{W}}

\title[$Q$-operator for the hyperbolic Calogero--Moser system]{Baxter $Q$-operator for the hyperbolic Calogero--Moser system}

\author{Martin Halln\"as}
\email{hallnas@chalmers.se}
\address{Department of Mathematical Sciences, Chalmers University of Technology and the University of Gothenburg, SE-412 96 Gothenburg, Sweden}

\thanks{Supported by the Swedish Research Council (Project-id 2018-04291).}

\date{\today}

\begin{abstract}
We introduce a $Q$-operator $\mathcal{Q}_z$ for the hyperbolic Calogero--Moser system as a one-parameter family of explicit integral operators. We establish the standard properties of a $Q$-operator, i.e.~invariance of Hamiltonians, commutativity for different parameter values and that its eigenvalues satisfy an explicitly given first order ordinary difference equation in the parameter $z$.
\end{abstract}

\maketitle

\tableofcontents

\section{Introduction}
Baxter \cite{Bax72,Bax82} first introduced the notion of a $Q$-operator as a technical tool in his study of the eight vertex model, allowing him to deduce Bethe-like equations for its eigenvalues even though a Bethe ansatz for the eigenvectors was lacking. A similar approach was later developed for the periodic Toda chain, first by Gutzwiller \cite{Gut81}, who could handle the $N=2,3$ and $4$-particle cases. After important simplifications by Sklyanin \cite{Skl85}, Pasquier and Gaudin \cite{PG92} proposed a $Q$-operator for the periodic Toda chain, given as an explicit integral operator, and used it to generalise Gutzwiller's Bethe equations for the energy spectrum to all particle numbers $N$. By now, $Q$-operators have been obtained for a number of models.

To be more specific, let us consider a quantum integrable $N$-particle system given by $N$ independent and pairwise commuting partial differential (or difference) operators $H_r$, $r=1,\ldots,N$. A corresponding $Q$-operator $\mathcal{Q}_z$ should depend on a parameter $z$ and its characteristic properties usually include (cf.~Kuznetsov and Sklyanin \cite{KS98})
\begin{enumerate}
\item invariance of Hamiltonians,
$$
\lbrack \mathcal{Q}_z,H_r\rbrack=0;
$$

\item commutativity,
$$
\lbrack \mathcal{Q}_z,\mathcal{Q}_w\rbrack=0;
$$

\item and the fact that its eigenvalues $\phi(z)$ on a joint eigenfunction of $H_r$ and $\mathcal{Q}_z$ satisfy an {\em ordinary} differential- or difference equation
\begin{equation*}
W(z,-\ii\hbar d/dz;\{E_r\})\phi(z) = 0,
\end{equation*}
involving the corresponding eigenvalues $E_r$ of $H_r$.
\end{enumerate}

In this paper, we focus on the $A_{N-1}$ hyperbolic Calogero--Moser system, which describes an arbitrary number of particles $N$ on the line that interact pairwise through the hyperbolic potential
$$
U(x_1,\ldots,x_N) = \sum_{1\leq i<j\leq N}u(x_i-x_j),\ \ \ u(x)\equiv 2g(g-\hbar)\mu^2\big/4\sinh^2(\mu x/2),
$$
where $\hbar\equiv h/2\pi>0$ is the reduced Planck constant, $g>0$ a coupling constant with dimension $\lbrack\text{action}\rbrack$ and $\mu>0$ a parameter with dimension $\lbrack\text{position}\rbrack^{-1}$. A complete set of independent pairwise commuting quantum integrals $H_1,\ldots,H_N$ are given by the following explicit formula:
\begin{multline}
\label{Hr}
H_r\equiv \frac{1}{(N-r)!}\sum_{0\leq s\leq [r/2]}\frac{1}{2^ss!(r-2s)!}\\
\cdot \sum_{\sigma\in S_N}\sigma\big(u(x_1-x_2)\cdots u(x_{2s-1}-x_{2s})\hat{p}_{2s+1}\cdots\hat{p}_r\big),
\end{multline}
with the momentum operators $\hat{p}_i\equiv -\ii\hbar\partial_{x_i}$, $i=1,\ldots,N$; see e.g.~\cite{OP83,Rui99,HR15}. In particular, the Schr\"odinger operator $H\equiv -\hbar^2\Delta+U$ can be obtained as the linear combination $H_1^2-2H_2$.

Introducing the hyperbolic weight function
\begin{equation}
\label{cWN}
\cW_N(g;x)\equiv \prod_{1\leq i<j\leq N}\big[4\sinh^2(\mu(x_i-x_j)/2)\big]^{g/\hbar},
\end{equation}
we recall that the quantum integrals $H_1,\ldots,H_N$ have asymptotically free joint eigenfunctions of the form
\begin{multline}
\label{PsiN}
\Psi_N((p_1,\ldots,p_N),g;(x_1,\ldots,x_N)) = \cW_N(g;(x_1,\ldots,x_N))^{1/2}\\
\cdot F_N((p_1/\hbar\mu,\ldots,p_N/\hbar\mu),g/\hbar;(\mu x_1,\ldots,\mu x_N)),
\end{multline}
with $F_N(u,\lambda;t)$ a function of the $2N+1$ dimensionless quantities
\begin{equation}
\label{ulam}
(u_1,\ldots,u_N)\equiv (p_1/\hbar\mu,\ldots,p_N/\hbar\mu),\ \ \lambda\equiv g/\hbar
\end{equation}
and
\begin{equation}
\label{t}
(t_1,\ldots,t_N)\equiv (\mu x_1,\ldots,\mu x_N)
\end{equation}
that is both analytic and $S_N$-invariant in $t$. (Here and below we take the positive square root of $\cW_N$.) The corresponding eigenvalue of $H_r$ is given by the $r$th symmetric function
$$
S_r(p)\equiv \sum_{1\leq i_1<\cdots<i_r\leq N}p_{i_1}\cdots p_{i_r}
$$
of the momenta $p_1,\ldots,p_N$.

In the case of  $N=2$ variables, we have $F_2(u,\lambda;t)=\exp((u_1+u_2)(t_1+t_2)/2)F(u_1-u_2,\lambda;t_1-t_2)$, where $F$ is essentially equal to the conical (or Mehler) function specialisation of the Gauss hypergeometric function ${}_2F_1$; see e.g.~Chapter 14 in \cite{Dig10}. In the general-$N$ case, the analog of $F$ was identified for the parameter values $\lambda=d/2$ with $1,2,4$ by Olshanetsky and Perelomov \cite{OP83} as the spherical function on $SL_n(\mathbb{F})/SU_n(\mathbb{F})$ with $\mathbb{F}=\mathbb{R},\mathbb{C},\mathbb{H}$. The appropriate generalisation of $F$ to arbitrary $\lambda>0$ (and $N\geq 3$) is the Heckman--Opdam hypergeometric function associated with the root system $A_{N-1}$, first constructed and studied by Heckman and Opdam \cite{HO87} in the context of an arbitrary root system. A corresponding generalised Fourier transform was introduced and developed by Opdam \cite{Opd95} as well as Cherednik \cite{Che97}.

As a first result, we show that the joint eigenfunctions $\Psi_N$ \eqref{PsiN} satisfy a one-parameter family of integral equations, with integration kernels
\begin{equation}
\label{Qz}
Q_z(g;x,y)\equiv \exp\left(\ii \frac{z}{\hbar}\sum_{i=1}^N(x_i-y_i)\right)
\frac{\cW_N(g;x)^{1/2}\cW_N(g;y)^{1/2}}{\prod_{i,j=1}^N \big[2\cosh(\mu(x_i-y_j)/2)\big]^{g/\hbar}}
\end{equation}
and eigenvalues that are explicit as well. More precisely, choosing the Weyl chamber
$$
G_N\equiv \left\{x\in\mathbb{R}^N\mid x_1<x_2<\cdots<x_N\right\},
$$
we prove the following theorem.

\begin{theorem}
\label{Thm:IntEq}
Letting $z\in\mathbb{R}$ and $(p,g,x)\in \mathbb{R}^N\times [\hbar,\infty)\times \mathbb{R}^N$, we have
\begin{equation}
\label{IntEq}
\int_{G_N}Q_z(g;x,y)\Psi_N(p,g;y)dy = \phi_z(p,g)\Psi_N(p,g;x),
\end{equation}
with eigenvalue
\begin{equation}
\label{phi}
\phi_z(p,g) = \prod_{i=1}^N \frac{\Gamma\big(g/2\hbar+\ii (p_i-z)/\hbar\mu\big)\Gamma\big(g/2\hbar-\ii (p_i-z)/\hbar\mu\big)}{\mu\Gamma(g/\hbar)}.
\end{equation}
\end{theorem}

From the account in Sections 4.1.1--2 of \cite{HR12}, it is readily infered that $Q_z$ satisfies the so-called kernel identities
\begin{equation}
\label{QzIds}
\big(H_r(x)-H_r(-y)\big)Q_z(x,y) = 0,\ \ \ r = 1,\ldots,N.
\end{equation}
For the Schr\"odinger operator $H\equiv-\hbar^2\Delta+U$, such an identity was first established by Langmann \cite{Lan00} in the more general elliptic case. The generalisation to all (symmetric elliptic) quantum integrals was later obtained by Ruijsenaars \cite{Rui06}. Using the identities \eqref{QzIds} as well as the manifest formal self-adjointness of $H_r$, we show that the left-hand side of \eqref{IntEq} is a joint eigenfunction of the quantum integrals $H_r$ with the same eigenvalue as $\Psi_N$. Combining this result with a suitable uniqueness result for regular joint eigenfunctions, due to Heckman--Opdam in the (radial) gauge corresponding to $F_N$, we thus arrive at an integral equation of the form \eqref{IntEq}. The explicit expression \eqref{phi} for the eigenvalue is then deduced by computing the dominant asymptotic behaviour of the integral in \eqref{IntEq} deep inside the Weyl chamber $G_N$. We suspect that all $g\geq 0$ could be allowed in Thm.~\ref{Thm:IntEq}, but, just as in \cite{HR15}, we need the stronger assumption $g\geq\hbar$ in order to control all aspects of its proof.

For $z\in\mathbb{R}$, we define $\mathcal{Q}_z$ as a one-parameter family of integral operators on $L^2(G_N)$ by
\begin{equation}
\label{cQz}
(\mathcal{Q}_zf)(x)\equiv \int_{G_N}Q_z(g;x,y)f(y)dy.
\end{equation}
By combining Thm.~\ref{Thm:IntEq} with the pertinent generalised Fourier (or eigenfunction) transform, we readily infer that it has all three of the above Properties (1)--(3) of a $Q$-operator for the hyperbolic Calogero--Moser system. This constitutes our main result and its precise formulation now follows.

\begin{theorem}
\label{Thm:QOp}
Assuming $g\in [\hbar,\infty)$ and $z\in\mathbb{R}$, the operator $\mathcal{Q}_z$, as defined by \eqref{cQz}, is bounded and self-adjoint and satisfies the commutation relations
\begin{equation}
\label{QzHrCom}
\lbrack \mathcal{Q}_z,H_r\rbrack=0,\ \ \ r = 1,\ldots,N,
\end{equation}
and
\begin{equation}
\label{QzQwCom}
\lbrack \mathcal{Q}_z,\mathcal{Q}_w\rbrack=0,\ \ \ z,w\in\mathbb{R}.
\end{equation}
Furthermore, introducing the generating function
$$
E(\gamma;p)\equiv \prod_{i=1}^N(\gamma+p_i) = \sum_{r=0}^N \gamma^{N-r}S_r(p),
$$
we have the difference equation
\begin{equation}
\label{phizEq}
E(\ii\mu(\hbar-g/2)-z;p)\phi_{z-i\hbar\mu}(p) = (-1)^NE(\ii\mu g/2-z;p)\phi_z(p).
\end{equation}
\end{theorem}

In the $N=2$ case, the results in Thms.~\ref{Thm:IntEq}--\ref{Thm:QOp} (except the difference equation \eqref{phizEq} in the latter), as well as the main ideas behind their proofs, can all be extracted from the paper \cite{HR18}. See also the recent preprint by Belousov, Derkachov, Kharchev and Khoroshkin \cite{BDKK23c} for a detailed account of $Q$-operators for the $2$-particle hyperbolic Calogero--Moser system.

While the $2$-variable case can be handled using results on classical conical functions, our arguments in this paper rely on the relatively recent theory of Heckman--Opdam hypergeometric functions. On the other hand, in \cite{HR15} we were able to obtain analogous results also for the $N=2$ hyperbolic relativistic Calogero--Moser (or Ruijsenaars--Schneider) system. We believe that the approach presented in the present paper can be used to generalize these results on the relativistic case to all particle numbers $N>3$, but, at the time of writing, a suitable uniqueness result for the pertinent joint eigenfunctions is missing.

In a remarkable paper, Belousov et al.~\cite{BDKK23a} recently obtained a $Q$-operator for the arbitrary-$N$ hyperbolic relativistic Calogero--Moser system; and, in the follow-up paper \cite{BDKK23b}, they established an integral equation analogous to \eqref{IntEq} for the joint eigenfunctions of the corresponding quantum integrals constructed via an explicit recursive scheme by the author and Ruijsenaars in \cite{HR14}.

The results of Belousov et al.~are obtained using methods that are distinctly different from the ones used in this paper. Specifically, their proof of commutativity of $Q$-operators (i.e.~Property (2)) proceeds in a more direct manner using a hypergeometric identity generalising \eqref{QzIds}; and the integral equation in question is deduced from the explicit recursive construction of the joint eigenfunctions.

We believe it would be interesting to further develop and compare the different methods indicated above. In particular, while Belousov et al.~already have obtained striking results at the relativistic level, the approach used here does not rely on the recursive construction of joint eigenfunctions, which is available only in the $A_{N-1}$-case, and could thus potentially be used to obtain interesting results on Calogero--Moser systems associated with other root systems.

The plan of the paper is as follows. In Section \ref{Sec:JEigFuncs}, we prepare the ground for our proofs of Thms.~\ref{Thm:IntEq}--\ref{Thm:QOp}. Specifically, we recall the connection between the quantum integrals $H_1,\ldots,H_N$ and the $A_{N-1}$-instance of the Heckman--Opdam hypergeometric system of PDEs as well as the asymptotic behaviour of the corresponding hypergeometric function. Section \ref{Sec:IntEq} contains the proof of Thm.~\ref{Thm:IntEq} and Section \ref{Sec:QOp} is devoted to the proof of Thm.~\ref{Thm:QOp}.

\section{Joint eigenfunctions}
\label{Sec:JEigFuncs}
In this section, we briefly review properties of the joint eigenfunctions $\Psi_N$ \eqref{PsiN} that we rely on in our proofs of Thms.~\ref{Thm:IntEq}--\ref{Thm:QOp}. It will be convenient to work in the (radial) $F_N$-gauge and with the dimensionless quantities \eqref{ulam}--\eqref{t}.

Therefore, we consider the PDOs
\begin{equation}
\label{Dr}
D_r(\lambda;t)\equiv (\hbar\mu)^{-r}\left(\cW_N^{-1/2}H_r\cW_N^{1/2}\right)\left(\hbar\lambda;\mu^{-1}t\right),\ \ \ r = 1,\ldots,N,
\end{equation}
which generate a commutative algebra of algebraically independent PDOs, containing, in particular, the second order PDO
\begin{equation}
\begin{split}
L_2 &\equiv \Delta+2\lambda\sum_{1\leq i<j\leq N}\coth\frac{t_i-t_j}{2}\left(\frac{\partial}{\partial t_i}-\frac{\partial}{\partial t_j}\right)\\
&= 2D_2-D_1^2-(\rho,\rho),
\end{split}
\end{equation}
where $(\cdot,\cdot)$ denotes the standard bilinear form on $\mathbb{C}^N$ and $\rho=\rho(\lambda)$ the Weyl vector for `multiplicity' $\lambda$, given by
$$
\rho = \frac{\lambda}{2}(N-1,N-3,\ldots,-N+3,-N+1).
$$

The system of PDEs
\begin{equation}
\label{hyperGSys}
D_r(t)F_N = S_r(u)F_N,\ \ \ r = 1,\ldots,N,
\end{equation}
essentially amounts to the Heckman--Opdam hypergeometric system associated with a root system of type $A_{N-1}$, multiplicity parameter $\lambda\geq 0$ and spectral parameter $u\in\mathbb{C}^N$.

We note that the orthogonal projection of $v\in\mathbb{R}^N$ onto the hyperplane $v_1+\cdots+v_N=0$ in $\mathbb{R}^N$ equals
$$
\pi(v)\equiv v-\frac{1}{N}(v,\underline{1})\underline{1},\ \ \ \underline{1}\equiv (1,\ldots,1).
$$
Up to a constant multiple, the hypergeometric system \eqref{hyperGSys} has a unique solution of the form
\begin{equation}
\label{FNF}
F_N(u,\lambda;t) = \exp\left(\frac{\ii}{N}(u,\underline{1})(t,\underline{1})\right)F(\ii\pi(u),\lambda;\pi(t))
\end{equation}
that is symmetric and analytic in $t$ on a suitable neighbourhood of the origin. (Here and below we extend $\pi$ by linearity to all of $\mathbb{C}^N$.) If we impose the normalisation condition $F(u,\lambda;0)=1$, the function $F$ is precisely Heckman and Opdam's hypergeometric function of type $A_{N-1}$ and, following the terminology of Brennecken and R\"osler \cite{BR23}, we refer to the corresponding function $F_N$ as the {\em extended} $A_{N-1}$ Heckman--Opdam hypergeometric function.

We note that, although the original uniqueness result of Heckman and Opdam \cite{HO87} (see also \cite{HO21}) on symmetric analytic solutions of hypergeometric systems concerns $F$, the above extension to $F_N$ follows as a straightforward corollary; see e.g.~Section 7 in \cite{HR15} and Section 2 in \cite{BR23}.

Taking $u\in\mathbb{R}^N$, we proceed to record the dominant asymptotics of the extended hypergeometric function $F_N(u,x)$ for
$$
m_N(x)\equiv \max_{i=1,\ldots,N-1}(x_i-x_{i+1})\to-\infty
$$
as well as a bound on the remainder, which exhibits its exponential decay; both of which are easily inferred from the asymptotic expansion of $F(\ii u,x)$ in Weyl chambers, as established by Heckman and Opdam \cite{HO87} (see also \cite{Opd95}). More precisely, introducing the dominant asymptotics function
\begin{equation}
\label{Fas}
F_N^{\mathrm{as}}(u,\lambda;x) := \sum_{\sigma\in S_N}c(-\sigma\ii u,\lambda)\exp((\sigma\ii u+\rho,x)),
\end{equation}
with the generalised Harish-Chandra $c$-function given by
$$
c(v,\lambda) = \frac{\tilde{c}(v,\lambda)}{\tilde{c}(\rho,\lambda)},\ \ \ \tilde{c}(v,\lambda) = \prod_{1\leq i<j\leq N}\frac{\Gamma(v_i-v_j)}{\Gamma(v_i-v_j+\lambda)},
$$
we have the following result.

\begin{proposition}
\label{Prop:FAs}
Let $\delta>0$ and assume that $\lambda\geq 0$ and that $u\in\mathbb{R}^N$ is regular, in the sense that $u_i-u_j\neq 0$ for all $1\leq i<j\leq N$. Then there exists a constant $C_\delta>0$ such that
\begin{equation}
\label{FBd}
|(F_N-F_N^{\mathrm{as}})(u,\lambda;x)| < C_\delta\exp\big((\rho,x)+m_N(x)\big)
\end{equation}
for all $x\in G_N$ with $m_N(x)<-\delta$.
\end{proposition}

\begin{proof}[Sketch of proof]
Specialising Thm.~6.3 in \cite{Opd95} to a root system of type $A_{N-1}$, we infer the following information on the asymptotic behaviour of $F_N$ in the Weyl chamber $G_N$:
$$
F_N(u,\lambda;x) = \sum_{\sigma\in S_N}c(-\ii\sigma u,k)\phi(\ii\sigma u+\rho,\lambda;x),
$$
where the generalized Harish-Chandra series $\phi$ is of the form
$$
\phi(\xi+\rho,\lambda;x) = \exp((\xi+\rho,x))\sum_{\chi\in Q_A^+}\Delta_\chi(\xi,\lambda)\exp((\chi,x)),
$$
with $Q_A^+$ being the $\mathbb{Z}_+$-span of the simple roots $e_i-e_{i+1}$, $i=1,\ldots,N-1$, and the coefficients $\Delta_\chi$ satisfying, in particular, the following properties: $\Delta_0\equiv 1$ and, for each $x_0\in -G_N$, there exists $K_{x_0}>0$ such that
$$
|\Delta_\chi(\xi,\lambda)| \leq K_{x_0}\exp((\chi,x_0)),\ \ \ \chi\in Q_A^+,\ \ \xi\in \ii\mathbb{R}^N.
$$
Choosing $x_0\in -G_N$ such that $0<m_N(x_0)<\delta/2$, we have
$$
m_N(x_0+x)\leq m_N(x_0)+m_N(x) < -\delta/2
$$
whenever $m_N(x)<-\delta$. From these two bounds, it is now a straightforward exercise to deduce that
\begin{equation*}
\sum_{\chi\in Q_A^+\setminus\{\underline{0}\}}|\Delta_\chi(\xi,k)|\exp(\langle \chi,x))\leq C_\delta\exp(m_N(x))
\end{equation*}
for $m_N(x)<-\delta$, and where $C_\delta>0$ can be chosen to depend continuously on $\delta$, from which the proposition clearly follows.
\end{proof}

\section{Integral equation}
\label{Sec:IntEq}
We turn now to the proof of Thm.~\ref{Thm:IntEq}. In order to easily make use of results discussed in Section \ref{Sec:JEigFuncs}, we substitute \eqref{PsiN} in \eqref{IntEq}, remove the overall factor $W_N(g;x)^{1/2}$, switch to the dimensionless quantities \eqref{ulam}--\eqref{t} and
\begin{equation}
\label{xi}
\xi\equiv z/\hbar\mu.
\end{equation}
We are thus led to consider the integrand
\begin{equation}
\label{Iz}
I_\xi(u,\lambda;t,s)\equiv  \exp\left(\ii \xi\sum_{i=1}^N(t_i-s_i)\right)
K_N(\lambda;t,s) F_N(u,\lambda;s) W_N(\lambda;s),
\end{equation}
with kernel function
\begin{equation}
\label{KN}
K_N(\lambda;t,s)\equiv \prod_{i,j=1}^N\left\lbrack 2\cosh\frac{t_i-s_j}{2}\right\rbrack^{-\lambda}
\end{equation}
and weight function
\begin{equation}
\label{WN}
W_N(\lambda;s)\equiv \cW_N\big(\hbar\lambda;\mu^{-1}s\big) = \prod_{1\leq i<j\leq N}\left[4\sinh^2\frac{s_i-s_j}{2}\right]^\lambda.
\end{equation}

Given $(u,\lambda)\in\mathbb{R}^N\times(0,\infty)$, the extended hypergeometric function $F_N$ is known to satisfy the bound
\begin{equation}
\label{FBd2}
|F_N(u,\lambda;s)|\leq C\exp((\rho,s))\prod_{1\leq i<j\leq N}(1+s_j-s_i),\ \ \ s\in G_N,
\end{equation}
for some constant $C>0$; see Cor.~7 in \cite{Saw08}, Cor.~3.1 \& Thm.~3.1 in \cite{Sch08} or Thm.~6.3 in \cite{HR15}.
Rewriting the weight function according to
\begin{equation}
\label{WNExpr}
W_N(s) = e^{-2(\rho,s)}\prod_{1\leq i<j\leq N}(1-e^{s_i-s_j})^{2\lambda}
\end{equation}
and using, in addition, the elementary estimate
\begin{equation}
\label{coshEs}
\left|\cosh\frac{w}{2}\right|^{-\lambda}\leq C(\im w)\exp\left(-\frac{\lambda}{2}|\re w|\right),\ \ \ |\im w| < \pi,
\end{equation}
where $C$ is a continuous function on $(-\pi,\pi)$, as well as the fact that $|\rho_i|\leq \lambda(N-1)/2$, we thus infer the bound
\begin{equation}
\label{IBd}
|I_z(u,\lambda;t,s)|\leq C(\re t,\im t)\exp\left(-\frac{\lambda}{2}||s||_1\right)\prod_{1\leq i<j\leq N}(1+s_j-s_i),\ \ \ s\in G_N,
\end{equation}
with $C$ a continuous function on $\mathbb{R}^N\times (-\pi,\pi)^N$, and where $||s||_1\equiv |s_1|+\cdots+|s_N|$. Since the singularities of the kernel function $K_N(t,s)$ are located at
$$
t_i = s_j\pm\frac{\ii\pi}{2}(2n+1),\ \ \ i,j = 1,\ldots,N,\ \ n\in\mathbb{N},
$$
it follows that the function
\begin{equation}
\label{scrF}
\mathscr{F}_\xi(u,\lambda;t)\equiv \int_{G_N}I_\xi(u,\lambda;t,s)ds,\ \ \ (u,\lambda,t)\in\mathbb{R}^N\times(0,\infty)\times\mathbb{R}^N,
\end{equation}
is well defined and extends to a holomorphic function of $t$ for $|\im t_i|<\pi$, $i=1,\ldots,N$. (Indeed, using Cauchy's integral formula and the uniform bound \eqref{IBd}, it is readily seen that we may differentiate in $t$ under the integral sign.)

From the kernel identities \eqref{QzIds}, we now infer that $\mathscr{F}_\xi(u,\lambda;t)$ is a solution to the system of (extended) hypergeometric PDEs \eqref{hyperGSys}. It is at this point that we need to assume that $\lambda\geq 1$, or equivalently that $g\geq \hbar$.

\begin{lemma}
Let $\lambda\geq 1$ and $u\in\mathbb{R}^N$. For $t\in\mathbb{C}^N$ with $|\im t_i|<\pi$, $i=1,\ldots,N$, we have the joint eigenfunction property
\begin{equation}
\label{scrFEigProp}
D_r(t)\mathscr{F}_\xi(u;t) = S_r(u)\mathscr{F}_\xi(u;t),\ \ \ r = 1,\ldots,N.
\end{equation}
\end{lemma}

\begin{proof}
If we substitute $\rho(\re\lambda)$ for $\rho$ in \eqref{FBd2}, Thm.~6.3 in \cite{HR15} implies that the resulting bound holds true as long as $\re\lambda>1$. It follows that $\mathscr{F}_\xi(u,\lambda;t)$ is analytic in $\lambda$ for $\re\lambda>1$, so that it suffices to prove the lemma for $\lambda>2$, say.

Using \eqref{PsiN}, \eqref{Qz}, \eqref{Dr}, \eqref{Iz} and \eqref{scrF}, we rewrite the left-hand side of \eqref{scrFEigProp} as
$$
(\hbar\mu)^{-r}W_N(\lambda,t)^{-1/2}\int_{G_N}H_r(g,x)Q_z(g;x,y)\Psi_N(p,g;y)dy
$$
cf.~\eqref{ulam}--\eqref{t}. Invoking the kernel identity \eqref{QzIds}, we obtain
$$
(\hbar\mu)^{-r}W_N(\lambda,t)^{-1/2}\int_{G_N}\Psi_N(p,g;y)H_r(g,-y)Q_z(g;x,y)dy.
$$
We note that $\lambda>2$, or equivalently $g>2\hbar$, ensures that $W(y)^{1/2}H_r(-y)W(y)^{1/2}$ is contained in $C(\mathbb{R}^N)$. Indeed, $\cW_N(g,y)^{1/2}$ is clearly in $C^2(\mathbb{R}^N)$, each of its factors is differentiated at most twice and each pole of $u(y_i-y_j)$, $1\leq i\neq j\leq N$, is matched by a corresponding zero (of order $>2$) of $\cW_N(g,y)^{1/2}$. Since $H_r(g,-y)$ is manifestly self-adjoint and the remaining factors in both $\Psi_N(p,g;y)$ and $Q_z(g;x,y)$ are smooth, we can thus use integration by parts to get
$$
(\hbar\mu)^{-r}W_N(\lambda,t)^{-1/2}\int_{G_N}Q_z(g;x,y)H_r(g,y)\Psi_N(p,g;y)dy,
$$
where $-y$ has been replaced by $y$ due to the absence of complex conjugation. Finally, appealing to the eigenvalue property $H_r(y)\Psi_N(p;y)=S_r(p)\Psi_N$, substituting \eqref{PsiN} and \eqref{Qz} and reverting back to the dimensionless quantities \eqref{ulam}--\eqref{t}, we arrive at the right-hand side of \eqref{scrFEigProp}.
\end{proof}

Since $\mathscr{F}_\xi(u,\lambda;t)$ is manifestly symmetric in $t$ and the hypergeometric function $F_N(u,\lambda;t)$ spans the space of symmetric solutions to \eqref{hyperGSys} that are analytic at the origin, we can thus conclude that
\begin{equation}
\label{scrFzFN}
\mathscr{F}_\xi(u,\lambda;t) = \mu_\xi(u,\lambda)F_N(u,\lambda;t)
\end{equation}
for some function $\mu_\xi$, which remains to be determined. To this end, we proceed to compute the dominant asymptotics of $\mathscr{F}_\xi(u,\lambda;t)$ for $m_N(t)\to -\infty$, which, when compared with \eqref{Fas}, will yield an explicit expression for $\mu_\xi(u,\lambda)$.

Using the expression \eqref{WNExpr} for $W_N$ and employing the bound \eqref{FBd}, a simple telescoping argument yields the estimate
\begin{equation}
\label{WFEst}
W_N(s)F_N(u;s) = e^{-(\rho,s)}\Big[e^{-(\rho,s)}F_N^{\mathrm{as}}(u,s)+O\big(e^{m_N(s)}\big)\Big],\ \ \ m_N(s)\to-\infty.
\end{equation}
If we introduce the function
\begin{equation*}
\begin{split}
\widetilde{K}_N(t,s) &\equiv \prod_{1\leq i\neq j\leq N}\left\lbrack 2\cosh\frac{t_i-s_j}{2}\right\rbrack^{-\lambda}\\
&= \frac{e^{(\rho,t+s)}}{\prod_{1\leq i<j\leq N}(1+e^{s_i-t_j})^\lambda(1+e^{t_i-s_j})^\lambda}
\end{split}
\end{equation*}
and take $\alpha\in(0,1)$, we have
\begin{equation}
\label{KtEst}
\widetilde{K}_N(t,s) = e^{(\rho,t+s)}\lbrack 1+R_N(t,s)\rbrack,
\end{equation}
with the remainder $R_N$ satisfying the bound
\begin{equation}
\label{RNBd}
|R_N(t,s)| < Ce^{(1-\alpha)m_N(t)},
\end{equation}
for all $t,s\in G_N$ such that $\max_{i=1,\ldots,N}|s_i-t_i|<-\alpha m_N(t)$, where $C>0$ is independent of $\alpha$. Since the remaining factors in $K_N(t,s)$ decay exponentially as $|s_i-t_i|\to\infty$, this state of affairs suggests that the dominant $m_N(t)\to-\infty$ asymptotics of $\mathscr{F}_\xi(u;t)$ is obtained by performing the substitutions $W(s)F_N(u;s)\to e^{-2(\rho,s)}F_N^{\mathrm{as}}(u;s)$ and $K_N(t,s)\to e^{(\rho,t+s)}\prod_{i=1}^N[2\cosh\frac{t_i-s_i}{2}\big]^{-\lambda}$ in $I_\xi$. In the following lemma, we make this suggestion precise.

\begin{lemma}
For $\lambda>0$, we have
\begin{multline}
\label{scrFAs}
\mathscr{F}_\xi(u,\lambda;t)
= e^{(\rho,t)+\ii\xi\sum_{i=1}^Nt_i}\\
\cdot \Bigg(\int_{\mathbb{R}^N}\frac{e^{-(\rho,s)-\ii\xi\sum_{i=1}^Ns_i}F_N^{\mathrm{as}}(u,\lambda;s)}{\prod_{i=1}^N\left[2\cosh\frac{t_i-s_i}{2}\right]^{\lambda}}ds+O\big(e^{rm_N(t)/2}\big)\Bigg),
\end{multline}
as $m_N(t)\to-\infty$, with decay rate
$$
r = \min(1,\lambda/4).
$$
\end{lemma}

\begin{proof}
For $t\in G_N$, we consider the domain
$$
D_N(t)\equiv \left\{s\in\mathbb{R}^N\mid |s_i-t_i|<-m_N(t)/4,\, i=1,\ldots,N\right\}\subset G_N.
$$
Note that $m_N(s)<m_N(t)/2<0$ whenever $t\in G_N$ and $s\in D_N(t)$. From \eqref{WFEst} and \eqref{KtEst}--\eqref{RNBd}, we can thus infer
\begin{multline*}
\int_{D_N(t)}I_\xi(u;t,s)ds\\
= e^{(\rho,t)+\ii\xi\sum_{i=1}^Nt_i}\Bigg(\int_{D_N(t)}\frac{e^{-(\rho,s)-\ii\xi\sum_{i=1}^Ns_i}F_N^{\mathrm{as}}(u;s)}{\prod_{i=1}^N\left[2\cosh\frac{t_i-s_i}{2}\right]^\lambda} dy+O\big(e^{m_N(t)/2}\big)\Bigg)
\end{multline*}
as $m_N(t)\to-\infty$.

As demonstrated after the lemma, it is desirable to arrive at the integral in the right-hand side over all of $\mathbb{R}^N$, since it can be evaluated explicitly. To this end, we observe that the function $e^{-(\rho,s)}F_N^{\mathrm{as}}(u;s)$ is bounded for $s\in\mathbb{R}^N$. Moreover, if $s\in\mathbb{R}^N\setminus D_N(t)$, then $|s_i-t_i|\geq -m_N(t)/4$ for at least one $i=1,\ldots,N$ and, by the elementary estimate \eqref{coshEs}, we clearly have
\begin{equation*}
\int_{t_i\mp m_N(t)/4}^{\pm\infty} \frac{ds_i}{\left[2\cosh\frac{t_i-s_i}{2}\right]^{\lambda}} = O\big(e^{\lambda m_N(t)/8}\big)
\end{equation*}
as $m_N(t)\to-\infty$. From these observations the lemma readily follows.
\end{proof}

At this point we can invoke the Fourier transform formula
\begin{equation}
\int_{\mathbb{R}}\frac{e^{ivw}}{\left\lbrack 2\cosh\frac{w}{2}\right\rbrack^\lambda}dw = \frac{\Gamma(\lambda/2+iv)\Gamma(\lambda/2-iv)}{\Gamma(\lambda)},
\end{equation}
which is easily inferred from a standard integral representation for the Beta function, cf.~Eq.~(5.12.7) in \cite{Dig10}. Indeed, when combined with \eqref{Fas}, it yields
\begin{multline*}
\int_{\mathbb{R}^N}\frac{e^{-(\rho,s)-\ii\xi\sum_{i=1}^Ns_i}F_N^{\mathrm{as}}(u,\lambda;s)}{\prod_{i=1}^N\left[2\cosh\frac{t_i-s_i}{2}\right]^{\lambda}}ds\\
= \sum_{\sigma\in S_N}c(-i\sigma u,\lambda)\prod_{i=1}^N \int_{\mathbb{R}}\frac{e^{\ii\lbrack(\sigma u)_i-\xi\rbrack s_i}}{\left[2\cosh\frac{t_i-s_i}{2}\right]^{\lambda}}ds_i\\
= e^{-(\rho,t)-\ii\xi\sum_{i=1}^Nt_i}F_N^{\mathrm{as}}(u,\lambda;t)\prod_{i=1}^N \frac{\Gamma(\lambda/2+\ii (u_i-z))\Gamma(\lambda/2-\ii (u_i-z))}{\Gamma(\lambda)}.
\end{multline*}
Substituting this expression in \eqref{scrFAs}, it becomes clear from Prop.~\ref{Prop:FAs} that the function $\mu_\xi$ in \eqref{scrFzFN} is given by
$$
\mu_\xi(u,\lambda) = \prod_{i=1}^N \frac{\Gamma(\lambda/2+\ii (u_i-z))\Gamma(\lambda/2-\ii (u_i-z))}{\Gamma(\lambda)}.
$$

Multiplying \eqref{scrFzFN} by $W_N(\lambda;t)^{1/2}$, substituting the above expression for $\mu_\xi$ and rewriting the resulting equation in terms of the momenta $p$ and $z$, coupling constant $g$ and positions $x$ and $y$ (cf.~\eqref{ulam}--\eqref{t} and \eqref{xi}), we arrive at the integral equation for the joint eigenfunctions $\Psi_N$ given by \eqref{IntEq}--\eqref{phi}. This concludes the proof of Thm.~\ref{Thm:IntEq}.

\section{$Q$-operator}
\label{Sec:QOp}
Introducing the renormalised joint eigenfunctions (cf.~\eqref{PsiN})
\begin{multline*}
\widehat{\Psi}_N((p_1,\ldots,p_N),g;(x_1,\ldots,x_N))\\
\equiv \widehat{\cW}_N(g;p)^{1/2}\cdot F_N((p_1/\hbar\mu,\ldots,p_N/\hbar\mu),g/\hbar;(\mu x_1,\ldots,\mu x_N))\cdot \cW_N(g;x)^{1/2},
\end{multline*}
where
$$
\widehat{\cW}_N(g;p)\equiv 1\big/\widehat{C}(g;p)\widehat{C}(g;p),\ \ \ \widehat{C}(g;p)\equiv \prod_{1\leq i<j\leq N}\frac{\Gamma(\ii (p_i-p_j)/\hbar\mu)}{\Gamma(g/\hbar+\ii (p_i-p_j)/\hbar\mu)},
$$
we continue with the proof of Thm.~\ref{Thm:QOp}. Specifically, generalising the treatment of the $N=2$ case in \cite{HR18}, we shall make the Hilbert space properties of the $Q$-operator $\mathcal{Q}_z$ plain by using the fact that the generalised Fourier transform
$$
\mathcal{F}_N(g): C_0^\infty(G_N)\subset L^2(G_N)\to L^2(G_N),\ \ \ g\geq 0,
$$
defined by
$$
(\mathcal{F}_N(g)f)(x)\equiv \frac{1}{h^{N/2}}\int_{G_N}\widehat{\Psi}_N(p,g;x)f(p)dp,
$$
extends to a unitary operator on $L^2(G_N)$. The validity of this claim is readily inferred from the product structure \eqref{FNF} of the extended hypergeometric function $F_N$ and the analogous result for the $A_{N-1}$-instance of the hypergeometric Fourier transform, first introduced and developed by Opdam \cite{Opd95} as well as Cherednik \cite{Che97}; see also e.g.~\cite{HO21}.

Multiplying \eqref{IntEq} by $\widehat{\cW}_N(g;p)^{1/2}f(p)$, $f\in C_0^\infty(G_N)$, and integrating over $p\in G_N$, we may change the order of integration in the left-hand side to obtain
$$
\mathcal{Q}_z(\mathcal{F}_N(f)) = \mathcal{F}_N(\phi_zf).
$$
From the well-known Gamma function properties (cf.~Chapter 5 in \cite{Dig10})
$$
\overline{\Gamma(z)} = \Gamma(\bar{z}),\ \ \ |\Gamma(s+it)|\leq |\Gamma(s)|,
$$
it is clear that that the eigenvalue $\phi_z(p)$ is a real-valued bounded function for $p\in\mathbb{R}^N$. Since $C_0^\infty(G_N)$ is dense in $L^2(G_N)$, boundedness and self-adjointness of $\mathcal{Q}_z$ as well as \eqref{QzQwCom} clearly follow.

For $r=1,\ldots,N$, the multiplication operator $f\mapsto S_rf$ is self-adjoint when equipped with the domain
$$
\mathrm{Dom}(S_r)\equiv \{f\in L^2(G_N)\mid S_rf\in L^2(G_N)\}.
$$
Each formal PDO $H_r$ can thus be promoted to a self-adjoint operator $\mathcal{F}S_r\mathcal{F}^*$, whose domain $\mathrm{Dom}(H_r)=\mathcal{F}(\mathrm{Dom}(S_r))$. Since $\phi_z$ is bounded, it is clear that $\mathcal{Q}_z\mathrm{Dom}(H_r)\subset \mathrm{Dom}(H_r)$ and that $Q_zH_rf=H_rQ_zf$ for each $f\in \mathrm{Dom}(H_r)$, which amounts to \eqref{QzHrCom}.

Finally, the difference equation \eqref{phizEq} is straightforward to establish by a direct computation using the Gamma function recurrence $\Gamma(z+1)=z\Gamma(z)$.

\bibliographystyle{amsalpha}

\end{document}